\documentclass[11pt]{article}
\usepackage{fullpage}
\usepackage[sc]{mathpazo}
\usepackage{amsmath,amssymb,amsthm,mathtools,enumerate,multirow,url}
\usepackage{clrscode3e}
\newtheorem{theorem}{Theorem}[section]
\newtheorem{proposition}{Proposition}[section]
\newtheorem{lemma}[theorem]{Lemma}

\theoremstyle{definition}

\newcommand{\diag}{\mathsf{diag}}
\newcommand{\diam}{\mathsf{diam}}
\newcommand{\dmaxmin}{\mathfrak{d}}
\newcommand{\ImL}{\mathsf{Im}(L_G)}
\newcommand{\vol}{\mathsf{vol}}
\newcommand{\Real}{\mathbb R}
\newcommand{\Poisson}{\mathcal P}
\newcommand{\ceil}[1]{\left\lceil #1 \right\rceil}
\newcommand{\braket}[1]{\left\langle #1 \right\rangle}
\newcommand{\norm}[1]{\left\lVert #1 \right\rVert}
\newcommand{\snorm}[1]{\lVert #1 \rVert}

\newcommand{\poly}{\mathrm{poly}}

\newcommand{\mybar}[1]{\lambda}

\title{Solving Laplacian Systems in Logarithmic Space}
\author{
Fran{\c c}ois Le Gall\\
Graduate School of Informatics\\
Kyoto University\\ 
\url{legall@i.u-kyoto.ac.jp}}

\begin{document}
\date{}
\maketitle
\thispagestyle{empty}
\setcounter{page}{1}

\begin{abstract}
We investigate the space complexity of solving linear systems of equations. While all known deterministic or randomized algorithms solving a square system of $n$ linear equations in $n$ variables require $\Omega(\log^2 n)$ space, Ta-Shma (STOC 2013) recently showed that on a quantum computer an approximate solution can be computed in logarithmic space, giving the first explicit computational task for which quantum computation seems to outperform classical computation with respect to space complexity. In this paper we show that for systems of linear equations in the Laplacian matrix of graphs, the same logarithmic space complexity can actually be achieved by a classical (i.e., non-quantum) algorithm. More precisely, given a system of linear equations $Lx=b$, where~$L$ is the (normalized) Laplacian matrix of a graph on $n$ vertices and $b$ is a unit-norm vector, our algorithm outputs a vector $\tilde x$ such that $\norm{\tilde x -x}\le 1/\poly(n)$ and uses only $O(\log n)$ space if the underlying graph has polynomially bounded weights. We also show how to estimate, again in logarithmic space, the smallest non-zero eigenvalue of $L$.
\end{abstract}

\newpage

\section{Introduction}
\paragraph{Background.}
Ta-Shma showed a few years ago that several fundamental tasks in linear algebra, such as computing the eigenvalues of an $n\times n$ matrix $A$ or finding a solution to a linear system of equations $Ax=b$, can be solved approximately with polynomial precision on a quantum computer using $O(\log n)$ space~\cite{TaShmaSTOC13}. In comparaison, no $o(\log^2 n)$-space classical (i.e., non-quantum) algorithm is known for these problems: the best (in the space complexity setting) known classical algorithms that solve them exactly use $O(\log^2 n)$ space \cite{Berkovitz84,Borodin+82,Csanky+FOCS75} and nothing better is known for only approximating the solutions. This breakthrough was the first example of concrete computational tasks for which a quantum algorithm outperforms the best known classical algorithms in the standard space complexity setting. The power of space-bounded quantum algorithms for matrix problems comes from their ability to represent an $n$-dimensional vector using $O(\log n)$ quantum bits of memory and consequently estimate the entries of $A^k$, given an $n\times n$ matrix~$A$ and an integer $k\le\poly(n)$, in logarithmic quantum space. In comparison, the best known algorithm for this task 
uses $O(\log^2 n)$ space \cite{Borodin+82}.

 Doron and Ta-Shma \cite{Doron+ICALP15} recently showed how to ``dequantize'' the algorithm \cite{TaShmaSTOC13} for classes of matrices where $A^k$ can be computed space efficiently, and in particular when $A$ is the adjacency matrix of a graph, in which case entries of $A^k$ can be estimated by a classical algorithm running~$k$ steps of the random walk corresponding to $k$. They focused on one task, computation of the eigenvalues, and succeeded in constructing a $O(\log n)$-space classical algorithm that computes the eigenvalues of such a matrix $A$, but only with constant precision (a much weaker accuracy that the polynomial precision obtained in the quantum case).

    
\paragraph{Our results.}
In this paper we investigate the classical (i.e., non-quantum) space complexity of matrix problems for matrices associated with graphs. More precisely, we focus on the Laplacian matrices of undirected weighted graphs (defined in Section \ref{sec:prelim}).  

The following definition will be useful to state our results concisely: We say that an undirected weighted graph $G$ on $n$ vertices has polynomially bounded weights if the weight of any edge is upper bounded by a polynomial of $n$ and lower bounded by a polynomial of $1/n$. Unweighted graphs, where the weight of each edge is one, are examples of such graphs.

Our first result is stated in the following theorem.
\begin{theorem}\label{th:main1}
Let $G$ be an undirected weighted graph on $n$ vertices with polynomially bounded weights. For any $\epsilon,\gamma\in(0,1)$, an $\epsilon$-additive approximation of the solution of a Laplacian system corresponding to~$G$ can be computed with probability at least $1-\gamma$ in $O(\log (n/\epsilon)+\log\log(1/\gamma))$ space.
\end{theorem}
Theorem \ref{th:main1} shows that if $G$ has polynomially bounded weights, linear systems of equations involving the Laplacian matrix of $G$ can be solved approximately with polynomial precision in $O(\log n)$ space, which gives a classical algorithm with the same space complexity and the same precision as Ta-Shma's quantum algorithm \cite{TaShmaSTOC13}. Note that Laplacian systems are a natural and very well-studied subclass of linear systems, and have a multitude of algorithmic applications (see \cite{Vishnoi13} for a survey). While in the time complexity setting extremely fast algorithms for Laplacian systems have been obtained in the past decade \cite{Kelner+STOC13,Peng+STOC14,Spielman+STOC04,Spielman+11,Spielman+13,Spielman+14}, we are not aware of any work on the space complexity of this problem. 

Roughly speaking, Theorem \ref{th:main1} is proved by showing that solving a Laplacian system reduces to computing powers of the normalized adjacency matrix of the graph, and then computing space-efficiently the entries of the powers by using random walks. We note that Chung and Simpson \cite{Chung+15} considered a fairly similar approach based on random walks to design time-efficient algorithms solving Laplacian systems under boundary conditions. In the present paper the focus is different since we consider space complexity. Indeed, the main difficulty here is to show that the reduction to random walks can be implemented space efficiently. 

The values of the eigenvalues of the Laplacian of a graph are related to many important properties of the graph --- these relations are precisely the main subject of the field of spectral graph theory. The smallest non-zero eigenvalue is especially a fundamental quantity. If the graph is connected, then this eigenvalue is called the algebraic connectivity of the graph and controls for example expanding properties of the graph. Our second result, stated in the following theorem, shows that the smallest non-zero eigenvalue of the Laplacian matrix of an undirected weighted graph with polynomially bounded weights can be approximated with constant multiplicative precision in logarithmic space.
\begin{theorem}\label{th:main2}
Let $G$ be an undirected weighted graph on $n$ vertices with polynomially bounded weights. For any $\delta,\gamma\in(0,1)$, a $\delta$-multiplicative approximation of the smallest non-zero eigenvalue of the Laplacian matrix of $G$ can be computed with probability at least $1-\gamma$ in $O(\frac{1}{\delta}\log n+\log\log(1/\gamma))$ space.
\end{theorem}

This result is obtained by developing a space-efficient version of a well known technique, called the power method and used for estimating the largest eigenvalue of a matrix, and showing how to implementing it space-efficiently using random walks.

To our knowledge the only related prior work on the space complexity of approximating the algebraic connectivity of a graph is the recent work by Doron and Ta-Shma \cite{Doron+ICALP15} mentioned above. Their techniques can be used to estimate the algebraic connectivity in logarithmic space with constant additive precision. Our techniques give constant multiplicative precision, which is a much stronger result (the smallest non-zero eigenvalue is always smaller than two; in many applications this eigenvalue is actually close to zero).

%

\section{Preliminaries}\label{sec:prelim}
\noindent{\bf General notations.}
In this paper, $\log$ always denote the natural logarithm.  We use $\Poisson_{s}$ to represent the Poisson distribution with parameter $s$ (for any real number $s>0$). This discrete probability distribution is defined as $\Poisson_{s}(k)=\frac{e^{-s}k^s}{k!}$ for any integer $k\ge 0$. 

For any real number $a<b$, the notation $(a,b)$, respectively $(a,b]$, represents the set of real numbers $x$ such that $a<x< b$, respectively $a<x\le b$. Given two real numbers $a,b$ and a positive real number $\delta$, we say that $a$ is a $\delta$-additive approximation of $b$ if $|a-b|\le \delta$, and say that $a$ is a $\delta$-multiplicative approximation of $b$ if $|a-b|\le \delta|a-b|$.

We will generally work in the vector space $\Real^n$ for some positive integer $n$ (representing the number of vertices of the graph considered). Given a vector $v\in\Real^n$, we write $\norm{v}$ its Euclidean norm, and $v^t$ its transpose. Given two vectors $u,v\in\Real^n$, we write $\braket{u,v}$ their inner product. Given any $s\in\{1,\ldots,n\}$, we use $e_s$ to denote the vector in $\Real^n$ with $s$-th coordinate 1 and all other coordinates zero.

\noindent{\bf The Laplacian and its eigenvalues.} We now define the Laplacian of a graph and introduce more specific notations used through this paper. It will be more convenient for us to work with the normalized version of the Laplacian of a graph, due to its natural connections with random walks. We refer to~\cite{Chung97} for details on the normalized Laplacian and all the notions below.

Let $G=(V,E)$ be an undirected weighted graph (possibly with loops) on $n$ vertices, i.e., a graph with a weight function $w\colon V\times V\to \Real$ satisfying $w(i,j)=w(j,i)$ for all pairs of vertices $(i,j)\in V\times V$, $w(i,j)= 0$ if $(i,j)\notin E$ and $w(i,j) > 0$ if $(i,j)\in E$. The degree of a vertex $i\in V$, denoted $d_i$, is defined as 
$
d_i=\sum_{j\in V} w(i,j).
$
The degree matrix of the graph $G$ is the $n\times n$ diagonal matrix $D_G=\diag(d_1,\ldots,d_n)$.
The volume of $G$ is defined as 
$
\vol(G)=\sum_{\ell=1}^n d_\ell.
$
Finally, if the graph has no isolated vertex, i.e., the degree of each vertex is positive, let us define the  quantity
\begin{equation}\label{eq:defd}
\dmaxmin=\frac{\max_{i\in\{1,\ldots,n\}}d_i}{\min_{i\in\{1,\ldots,n\}}d_i}. 
\end{equation}

The (normalized) Laplacian of the graph $G$ is the $n\times n$ symmetric matrix $L_G$ such that
\[
L_G[i,j]=\left\{\begin{array}{ll}
1 -\frac{w(i,i)}{d_i}&\textrm{ if } i=j \textrm{ and }d_i\neq 0,\\
-\frac{w(i,j)}{\sqrt{d_id_j}} &\textrm{ if } (i,j)\in E,\\
0 &\textrm{ otherwise},\\
\end{array}
\right.
\]
for all $i,j\in\{1,\ldots,n\}$. Let $\lambda_1\le \lambda_2 \cdots\le \lambda_n$ denote its eigenvalues. It is known that $\lambda_1\ge 0$ and $\lambda_n\le 2$. Moreover, $\lambda_2\neq 0$ if and only if $G$ is a connected graph. We have the general upper bound $\lambda_2\le n/(n-1)$. When $G$ is connected, we obtain the lower bound 
\begin{equation}\label{eq:LBl2}
\lambda_2\ge \frac{1}{\diam(G)\vol(G)},
\end{equation}
 where $\diam(G)$ denotes the diameter of $G$ (the maximal distance between two vertices in $G$). 
Let $\{u_1,\ldots,u_n\}$ be an orthonormal basis for $L_G$, where $u_i$ is an eigenvector associated with eigenvalue $\lambda_i$. We can thus write
$
L_G=\sum_{i=1}^n \lambda_i u_iu_i^t.
$
It is easy to check that
$
u_1=\frac{1}{\sqrt{\vol(G)}}\left(\sqrt{d_1},\ldots,\sqrt{d_n}\right)^t
.
$

\noindent{\bf Image of $L_G$ and its pseudo inverse.}
We define $\ImL$ as the image of $L_G$, that is, the linear span of the vectors $u_{c}$, $u_{c+1},\ldots, u_n$, where $c$ is the smallest integer such that $\lambda_c\neq 0$. The pseudo inverse (also called the Moore-Penrose pseudo inverse) of~$L_G$, denoted $L_G^\dagger$, is the matrix
\vspace{-2mm}
\[
L_G^\dagger=\sum_{i=c}^n \lambda_i^{-1} u_iu_i^t,
\]
\vspace{-3mm}

\noindent 
which corresponds to the inverse of $L_G$ when restricted to the subspace $\ImL$. Note that if the graph $G$ is connected then $c=2$. In that case $\ImL$ is simply the subspace of $\Real^n$ consisting of all vectors orthogonal to $u_1$.

\noindent{\bf The transition matrix and random walks.}
Another useful matrix is the transition matrix of $G$, which represents one step of a random walk on the graph $G$. Assume that the graph has no isolated vertex. Then the transition matrix of $G$ is the $n\times n$ matrix $P_G$ defined as
\[
P_G[i,j]=\left\{\begin{array}{ll}
w(i,j)/d_i &\textrm{ if } (i,j)\in E,\\
0 &\textrm{ otherwise}\\
\end{array}
\right.
\]
for all $i,j\in\{1,\ldots,n\}$.
Note that
$
L_G = I - D_G^{1/2}P_G D_G^{-1/2}.
$


\noindent{\bf Laplacian systems.}
Given a Laplacian $L_G$ of an undirected weighted graph $G$ and a vector $b\in\ImL$, the corresponding Laplacian system is the equation
$
L_Gx=b.
$
The goal is to solve this equation, i.e., compute $L_G^\dagger b$
or to compute an (additive or approximative) approximation of it. 
A standard assumption is that $G$ is connected since one can always reduce to this case dealing individually with each connected component (note that this reduction is space-efficient since one can identify the connected components of a graph in logarithmic space). For this reason, in Sections \ref{sec:main} and \ref{sec:l2} of this paper we will assume that $G$ is connected. For convenience, and without loss of generality, we will also assume that $b$ is a unit-norm vector.

\noindent{\bf Remark on numerical precision.}
The algorithms we present in this paper perform arithmetic operations on real numbers. When working with real numbers numerical precision is always a delicate issue, and especially in a space-bounded setting. To illustrate this point, let us consider as a simple example the multiplication of an integer~$a$ by the irrational numbers $\sqrt{2}$.
One of the most rigorous ways to proceed is to work in the bit complexity model, and compute only a fixed number of bits of $a\sqrt{2}$, which necessarily introduce a (small) approximation error. Analyzing rigorously a complicated algorithm in this model is typically extremely tedious, since the approximation errors of essentially all the arithmetic steps of the algorithm have to been considered. In this paper we (implicitly) use a slightly more abstract, but still fairly standard, model where basic arithmetic operations involving ``reasonable'' numbers can be implemented exactly. More precisely, we assume that the standard arithmetic operations (addition, substraction, multiplication, division) on $O(\log m)$-bit reals numbers (i.e., real numbers of absolute value between $m^{-c}$ and $m^{c'}$ for some constant $c$ and $c'$) can be done exactly in $O(\log m)$ space. This model enables us to focus on the most interesting and important algorithmic aspects of our approach, without having to deal with a multitude of minor technical details. Note that all the statements of the technical results of Sections \ref{sec:main} and \ref{sec:l2} hold with only minor modifications in the bit complexity model as well. The statements of Theorems \ref{th:main1} and \ref{th:main2} in the introduction, described in the more convenient setting of graphs with polynomially bounded weighted, hold without any modification even in the bit complexity model.
  
\section{Space-efficient Laplacian Solver}\label{sec:main}
In this section we describe how to compute in logarithmic space an approximate solution of the Laplacian system $L_Gx=b$, 
and prove Theorem \ref{th:main1}. We assume that $G$ is connected.

The theoretical foundation of our algorithm comes from the formula of the following theorem, which shows that the solution $L_G^\dagger b$ to the Laplacian system can be approximated using powers of the transition matrix $P_G$ of the graph.
\begin{theorem}\label{th:formula}
Let $\epsilon$ be any positive real number and $\lambda\in(0,2]$ be a lower bound on all the non-zero eigenvalues of $L_G$ (i.e., a lower bound on $\lambda_2$). For any positive integers $T\ge \frac{\log(6/(\epsilon\lambda))}{\lambda}$, $N\ge \frac{6T}{\epsilon}$ and $K\ge \max(6T,\log(6T/\epsilon))$, the inequality 
\[
\norm{\Big(L_G^\dagger-\frac{T}{N}\sum_{j=1}^{N}\sum_{k=0}^{K-1} 
\Poisson_{jT/N}(k)D_G^{1/2}P_G^kD_G^{-1/2}\Big)b}\le \frac{\epsilon\norm{b}}{2}
\]
holds for any vector $b\in \ImL$.
\end{theorem}
\begin{proof}
Observe that
$
D_G^{1/2}P_G^k D_G^{-1/2} = \sum_{i=1}^n(1-\lambda_i)^k u_i^tu_i
$
for any integer $k\ge 0$.
We show below that the inequality
\[
\left|
\frac{1}{\lambda_i}-\frac{T}{N}\sum_{j=1}^{N}\sum_{k=0}^{K-1} 
\Poisson_{jT/N}(k)(1-\lambda_i)^k 
\right|
\le \epsilon/2,
\]
holds
for all $i\in\{2,\ldots,n\}$, which will imply the statement in the theorem.

For any $i\in\{2,\ldots,n\}$, first observe that 
$
\frac{1}{\lambda_i}=\int_0^\infty e^{-t\lambda_i}dt.
$
For any $T\ge 0$, we thus have
\begin{equation}\label{eq1}
\left|\frac{1}{\lambda_i}-\int_0^T e^{-t\lambda_i}dt\right|\le \int_T^\infty e^{-t\lambda_i}dt=\frac{e^{-T\lambda_i}}{\lambda_i}.
\end{equation}
Let us approximate the integral by a right Riemann sum: for any integer $N\ge 1$,
\begin{equation}\label{eq2}
\left|\int_0^T e^{-t\lambda_i}dt - \frac{T}{N}\sum_{j=1}^N e^{-j\frac{T}{N}\lambda_i}\right|
\le 
\frac{T}{N} \times (1-e^{-T\lambda_i})\le \frac{T}{N}.
\end{equation}

For any $s\ge 0$ we have 
\begin{align*}
e^{-s\lambda_i}=e^{-s}e^{(1-\lambda_i)s}
=e^{-s}\sum_{k=0}^{\infty}\frac{((1-\lambda_i)s)^k}{k!}
=\sum_{k=0}^{\infty}\Poisson_s(k)(1-\lambda_i)^k.
\end{align*}
For any integer $K\ge 1$ we thus have
\begin{align*}
\left|e^{-s\lambda_i}-\sum_{k=0}^{K-1} \Poisson_s(k)(1-\lambda_i)^k\right|
\le\sum_{k=K}^\infty \Poisson_s(k)
\le 
\left(1-\frac{s}{K+1}\right)^{-1}\Poisson_s(K)
\le \left(1-\frac{s}{K+1}\right)^{-1}\frac{s^Ke^{-s}}{\sqrt{2\pi K}K^K e^{-K}}.
\end{align*}
The first inequality uses $\lambda_i\in (0,2]$.
The second inequality uses a well-known upper bound on the tail probability of the Poisson distribution \cite{Glynn87}. The third inequality uses a standard lower bound on $K!$ from Stirling approximation. For any positive integer $K\ge e^2 s$, where $e=2.718\ldots$ is Euler's number, and in particular for any $K\ge 6 s$,  we get 
\begin{equation}\label{eq3}
\left|e^{-s\lambda_i}-\sum_{k=0}^{K-1} \Poisson_s(k)(1-\lambda_i)^k\right|
\le \left(1-\frac{1}{e^2}\right)^{-1}\frac{e^{-s}}{\sqrt{2\pi K}e^K}\le 
\frac{1}{e^K}.
\end{equation}

Combining Inequalities (\ref{eq1}), (\ref{eq2}) and (\ref{eq3}), for any positive integer $K\ge 6 T$ we obtain:
\begin{align*}
\left|\frac{1}{\lambda_i}-\frac{T}{N}\sum_{j=1}^{N}\sum_{k=0}^{K-1} 
\Poisson_{jT/N}(k)(1-\lambda_i)^k 
\right|
&\le \frac{e^{-T\lambda_i}}{\lambda_i}+\frac{T}{N}+\frac{T}{e^K}.
\end{align*}
Taking values of $T,N$ and $K$ as in the statement of the theorem guarantees that the right side of the above inequality is at most $\epsilon/2$.
\end{proof}

We now present a lemma that shows how to approximate space-efficiently the Poisson distribution appearing in the formula of Theorem \ref{th:formula}. This is done by using the well-known property that the Poisson distribution can be expressed as the limit distribution of binomial random variables, and showing that the convergence is fast enough.
\begin{lemma}\label{lemma:Poisson}
For any $\delta\in(0,1)$ and any $\zeta>0$, there exists a $O(\log{(ks/\delta)+\log\log(1/\zeta)})$-space algorithm that outputs a $\delta$-additive approximation of $\Poisson_{s}(k)$ with probability at least $1-\zeta$.
\end{lemma}
\begin{proof}
Let $n$ be any integer such that $n\ge 2(k^2+s^2)/\delta$.
Let $X_n$ be the binomial random variable with parameters $n$ and $s/n$, i.e., $X_n$ is the number of successes in $n$ repeated trials of a binomial experiment with success probability $s/n$. Standard computations (see, e.g., page 99 of~\cite{Mitzenmacher+05}) show that
\[
\Poisson_{s}(k)(1-k/n)^k\left(1-\frac{s^2}{n}\right)
\le
\Pr[X_n=k]\le 
\Poisson_{s}(k)\left(\frac{1}{1-s k/n}\right).
\]
Using the inequality $1/(1-x)\le 1+2x$ valid for any $x\in(0,1/2)$ and the inequality $(1-x)^k\ge 1-kx$ valid for any $k\ge 0$, we get
\[
\Poisson_{s}(k)\left(1-\frac{k^2+s^2}{n}\right)
\le
\Pr[X_n=k]\le 
\Poisson_{s}(k)\left(1+\frac{2s k}{n}\right).
\]
With our value of $n$, and since $\Poisson_{s}(k)\le 1$, we obtain
\begin{equation}\label{ineq:Poisson}
\left|\Poisson_s(k)-\Pr[X_n=k]\right|\le \delta/2.
\end{equation}

Our algorithm for estimating $\Poisson_s(k)$ is as follows. We will use the value $n=\ceil{2(k^2+s^2)/\delta}$. The algorithm creates a counter $C$ initialized to zero, and repeat $m=\ceil{\frac{2\log(2/\zeta)}{\delta^2}}$ times the following: perform $n$ trials of a binomial experiment with success probability $s/n$ and increment~$C$ by one if exactly $k$ among these $n$ trials succeeded. The algorithm finally outputs $C/m$. 

The expected value of the output of this algorithm is precisely the quantity $\Pr[X_n=k]$ considered above. From Chernoff bound, the output of the algorithm is thus a $\delta/2$-additive approximation of this quantity with probability at least $1-2e^{-m\delta^2/2}\ge 1-\zeta$. Combining this with Inequality (\ref{ineq:Poisson}) and the triangular inequality, we conclude that the output of the algorithm is a $\delta$-additive approximation of $\Poisson_s(k)$ with the same probability.

Finally, observe that the space complexity of the algorithm is linear in $\log(1/\delta)$, $\log k$, $\log s$ and $\log\log (1/\zeta)$. 
\end{proof}

We are now ready to present our algorithm for estimating $L_G^\dagger b$. The algorithm, denoted Algorithm $\mathcal{A}$, is described in Figure \ref{fig:algorithmA} and analyzed in the following theorem.
\begin{figure}[ht!]
\begin{center}
\fbox{
\begin{minipage}{15 cm} \vspace{-4mm}
\begin{codebox}
\zi
Input: an undirected, weighted and connected graph $G$ on $n$ vertices, 
\zi
\hspace{11mm}
an integer $i\in\{1,\ldots,n\}$, a unit vector $b\in \ImL$
\zi
\hspace{11mm}  a precision parameter $\epsilon\in(0,1]$, an error parameter $\gamma\in(0,1]$
\zi
\li
$T\leftarrow\ceil{\frac{\log(6/(\epsilon\lambda))}{\lambda}}$;
$N\leftarrow \ceil{\frac{6T}{\epsilon}}$; $K\leftarrow\ceil{\max(6T,\log(\frac{6T}{\epsilon}))}$;
\li
$\delta\leftarrow\left(6TK \sqrt{\sum_{\ell=1}^n\frac{{d_{\ell}}}{{d_i}}}\right)^{-1}$; $\zeta\leftarrow \frac{\gamma}{NK(1+n)}$; $r\leftarrow \ceil{\frac{\log(2/\zeta)}{2\delta^2}}$;
\li
$R \leftarrow 0$;
\li
\For $j$ from 1 to $N$ \Do
\li
\For $k$ from 0 to $K-1$ \Do
\li
Compute an approximation $a$ of $\Poisson_{jT/N}(k)$ using Lemma \ref{lemma:Poisson} with $\delta$ and $\zeta$;
\li
\For $\ell$ from 1 to n \Do
\li
$S\leftarrow 0$;
\li
\Repeat\hspace{3mm} $r$ times:
\li
Run the walk $P_G$ starting on vertex $\ell$ for $k$ steps;
\li
If the walk ends on vertex $i$ then
$S\leftarrow S+1$;
\End
\li
$R\leftarrow R+\frac{ab_\ell S}{r}\times \sqrt{\frac{d_{\ell}}{d_i}}$;
\End
\End
\End
\li
Output $RT/N$.
\end{codebox}
\end{minipage}
}
\end{center} 
\caption{Algorithm $\mathcal{A}$ computing an $\epsilon$-additive approximation of the $i$-th entry of $L_G^\dagger b$ with probability at least $1-\gamma$.}\label{fig:algorithmA}
\end{figure}
\begin{theorem}\label{th:algA}
Let $G$ be an undirected weighted connected graph and
$\lambda\in(0,2]$ be a lower bound on the second smallest zero eigenvalue $\lambda_2$ of $G$.
Algorithm~$\mathcal{A}$ outputs an $\epsilon$-additive approximation of the $i$-th entry of $L_G^\dagger b$ with probability at least $1-\gamma$, and uses $O(\log\big(\frac{n\dmaxmin}{\epsilon\lambda}\big)+\log\log(1/\gamma))$ space.
\end{theorem}
\begin{proof}
Remember that for any $s\in\{1,\ldots,n\}$, $e_s$ denotes the $1\times n$ vector with $s$-th coordinate~1 and all other coordinates zero. Note that the $i$-th coordinate of $L_G^\dagger b$, which is the quantity we want to estimate, is $\braket{e_i,b^t L_G^\dagger}$ since $L_G^\dagger$ is symmetric. From Theorem \ref{th:formula} and the triangular inequality, we know that any $\frac{\epsilon}{2}$-additive approximation of 
\begin{align*}
\frac{T}{N}\sum_{j=1}^{N}\sum_{k=0}^{K-1} 
\Poisson_{jT/N}(k)\braket{e_i,b^tD_G^{1/2}P_G^kD_G^{-1/2}}
=
\frac{T}{N}\sum_{j=1}^{N}\sum_{k=0}^{K-1} \sum_{\ell=1}^n
\Poisson_{jT/N}(k)b_\ell\sqrt{\frac{{d_\ell}}{{d_i}}}\braket{e_i,e_\ell P_G^k}
\end{align*}
is an $\epsilon$-additive approximation of $\braket{e_i,b^t L_G^\dagger}$. We show below that Algorithm $\mathcal{A}$ precisely outputs an $\epsilon/2$-additive approximation of this quantity.

Note that the probability of a walk $P_G$ starting on vertex $\ell$ reaches vertex $i$ after exactly $k$ steps is $\langle e_i,e_\ell P_G^k\rangle$, the $i$-th coordinate of $e_\ell P_G^k$. At the end of Steps 9-11 we thus have
\begin{equation}\label{eq:Chernoff}
\Pr\left[|S/r-\langle e_i,e_\ell P_G^k\rangle|\le \delta\right]\ge 1-2e^{-2r\delta^2}\ge 1-\zeta,
\end{equation}
from Chernoff bound.
Lemma \ref{lemma:Poisson} also shows that $a$ is a $\delta$-additive approximation of $\Poisson_{jT/N}(k)$ with probability at least $1-\zeta$. Let us continue our analysis under the assumption that all these approximations are correct (we discuss the overall success probability at the end of the proof). At Step 12 we thus have
\[
\left|
\frac{a S}{r} -
\Poisson_{jT/N}(k)\braket{e_i,e_\ell P_G^k}
\right|\le 
\left(\delta+\Poisson_{jT/N}(k)+\braket{e_i,e_\ell P_G^k}\right)\delta\le 3\delta.
\]
The output of the algorithm at Step 13 then satisfies
\begin{align*}
\left|
\frac{RT}{N} - \frac{T}{N}\sum_{j=1}^{N}\sum_{k=0}^{K-1} \sum_{\ell=1}^n
\Poisson_{jT/N}(k)b_\ell\sqrt{\frac{{d_\ell}}{{d_i}}}\braket{e_i,e_\ell P_G^k}
\right|
\le 
\frac{T}{N}\sum_{j=1}^{N}\sum_{k=0}^{K-1} \sum_{\ell=1}^n
3\delta |b_\ell|\sqrt{\frac{{d_\ell}}{{d_i}}}
&=
TK \sum_{\ell=1}^n
3\delta |b_\ell|\sqrt{\frac{{d_\ell}}{{d_i}}}\\
&\!\le
3\delta TK \sqrt{\sum_{\ell=1}^n\frac{{d_\ell}}{{d_i}}}
\le \epsilon/2.
\end{align*}

The space complexity of the algorithm is $O(\log(n\dmaxmin/(\lambda\epsilon))+\log\log(1/\gamma))$, from Lemma \ref{lemma:Poisson} and the observation that only registers of this size are needed to implement the algorithm. Finally, let us discuss the success probability of this algorithm. Errors can only occur at Steps 6 or 10. From Inequality (\ref{eq:Chernoff}) and Lemma \ref{lemma:Poisson}, and using the union bound, we know that the overall success probability is at least $1-\zeta NK (1+n)= 1-\gamma$. 
\end{proof}
Theorem \ref{th:algA} implies Theorem \ref{th:main1} by observing that $\dmaxmin/\lambda$ can be upper bounded by a polynomial in $n$ when the weights are polynomially bounded, as shown in Equations (\ref{eq:defd}) and (\ref{eq:LBl2}) of Section~\ref{sec:prelim} (if $G$ is not connected we can simply apply Theorem \ref{th:l2} on each connected component).

\section{Space-efficient Approximation of the Spectral Gap}\label{sec:l2}

Let us consider the following matrix:
\[
M_G=\frac{1}{2}\left(I+D_G^{1/2}P_GD_G^{-1/2}\right).
\]
Note that $M_G$ is a symmetric matrix. Its eigenvalues are
$
0 \le 1-\lambda_n/2 < \cdots < 1-\lambda_2/2 < 1-\lambda_1/2 = 1.
$
The eigenvectors of $M_G$ are the same as the eigenvectors of $L_G$. In particular, the eigenvector of $M_G$ corresponding to the eigenvalue $1$ is $u_1$, and the eigenvector of $M_G$ corresponding to the eigenvalue $1-\lambda_2/2$ is $u_2$.
A well-known approach for approximating the largest eigenvalue of a matrix is the power method (see, e.g., \cite{Vishnoi13}). Our idea is to apply this method on $M_G$ restricted to $\ImL$, and compute the ratio $\snorm{M_G^{k+1}v}/\snorm{M_G^{k}v}$ on a random vector $v\in \ImL$ --- it is easy to show that with high probability this ratio is close to $1-\lambda_2$ for large enough $k$. In this section we will develop a space-efficient version of this approach, and prove Theorem \ref{th:main2}

When using the power method to estimate $1-\lambda_2/2$, we need to use vectors orthogonal to $u_1$ that have a non-negliglible "overlap" with the eigenvector $u_2$. We say that a vector $v\in \Real^n$ is good if the following three conditions are satisfied: $\norm{v}=1$, $v\in \ImL$, and $|\langle v,u_2\rangle|\ge \frac{1}{\sqrt{2}n\dmaxmin}$. While a random unit-norm vector in $\ImL$ is a good vector with high probability, several technical difficulties arise when considering space-efficient vector sampling. Instead of using such probabilistic arguments, we introduce below a set $\Sigma$ of vectors that necessarily contains at least one good vector.  

Let $\Sigma\subset\Real^n$ be the set containing the $n(n-1)/2$ vectors defined as follows. Each of these vectors corresponds to choosing two distinct indexes $i,j\in\{1,\ldots,n\}$ and taking the $n$-dimensional vector with $i$-th coordinate $-\frac{1}{\sqrt{1+d_i/d_j}}$, $j$-th coordinate $\frac{1}{\sqrt{1+d_j/d_i}}$, and all other coordinates being zero. The following easy lemma shows that $\Sigma$ indeed contains at least one good vector.

\begin{lemma}\label{lemma:good}
There exists a good vector in $\Sigma$.
\end{lemma}
\begin{proof}
Any vector in $\Sigma$ is a unit-norm vector orthogonal to $u_1$. We show below that there exists a vector in $\Sigma$ that also satisfies the third condition of the definition of good vectors. 

Let us write $u_2=(x_1,\ldots,x_n)$. Let $S^+\subseteq \{1,\ldots,n\}$ be the set of indices $\ell$ such that $x_\ell\ge 0$, and $S^-\subseteq \{1,\ldots,n\}$ be the set of indices $\ell$ such that $x_\ell< 0$. Since $u_2$ is orthogonal to $u_1$, we have 
$
\sum_{\ell\in S^+} \sqrt{d_\ell}|x_\ell| = \sum_{\ell\in S^-} \sqrt{d_\ell}|x_\ell|.
$ 
Since $u_2$ is a unit vector we have 
\[
1= \sum_{\ell\in S^+} x_\ell^2 + \sum_{\ell\in S^-} x_\ell^2 
\le
\sum_{\ell\in S^+} |x_\ell| + \sum_{\ell\in S^-} |x_\ell| . 
\]
We conclude that 
\[
\sum_{\ell\in S^+} \sqrt{d_\ell}|x_\ell| \ge \frac{\min_{\ell\in\{1,\ldots,n\}}\sqrt{d_{\ell}}}{2} \textrm{ and } \sum_{\ell\in S^-} \sqrt{d_\ell}|x_\ell| \ge \frac{\min_{\ell\in\{1,\ldots,n\}}\sqrt{d_{\ell}}}{2}
\]
which implies that there exist $i\in S^-$ and $j\in S^+$ such that 
$|x_{i}|$ and $|x_{j}|$ are at least $\frac{1}{2n\sqrt{\dmaxmin}}$. Let~$v$ be the vector in $\Sigma$ with $i$-th coordinate $-\frac{1}{\sqrt{1+d_i/d_j}}$ and $j$-th coordinate $\frac{1}{\sqrt{1+d_j/d_i}}$.
The inner product of $v$ and $u_2$ is thus at least 
\begin{align*}
\frac{1}{2n\sqrt{\dmaxmin}}\left(\frac{1}{\sqrt{1+d_i/d_j}}+\frac{1}{\sqrt{1+d_j/d_i}}\right)
&=
\frac{1}{2n\sqrt{\dmaxmin}}\left(\frac{d_j}{\sqrt{d_j+d_i}}+\frac{d_i}{\sqrt{d_i+d_j}}\right)
\ge \frac{1}{\sqrt{2}n\dmaxmin},
\end{align*}
as claimed.
\end{proof}

The following proposition is our version of the power method.
\begin{proposition}\label{proposition:power}
Let $\delta$ be any real number such that $0<\delta\le 1$, and $\zeta$ be any real number such that $0<\zeta\le \delta\lambda_2/12$.
\begin{itemize}
\item[(i)]
For any integer $k\ge 0$, any non-zero vector $v\in \ImL$ and any $\zeta$-multiplicative approximations $C_1$ and $C_2$ of $\snorm{M_G^{k}v}$ and $\snorm{M_G^{k+1}v}$, respectively, the inequality
$
(1-\delta)\lambda_2\le 2\left(1-C_2/C_1\right)
$ holds.
\item[(ii)]
Let $v$ be a good vector.
For any integer  
$
k\ge\frac{3\log(\sqrt{2}n\dmaxmin)}{\delta\lambda_2}-1
$
and any $\zeta$-multiplicative approximations $C_1$ and $C_2$ of $\snorm{M_G^{k}v}$ and $\snorm{M_G^{k+1}v}$, respectively, the inequality
$
2\left(1-C_2/C_1\right)\le (1+\delta)\lambda_2
$
holds
\end{itemize}
\end{proposition}
\begin{proof}
Let us first prove part (i).
We have
\begin{align*}
\frac{C_2}{C_1}&\le \frac{\norm{M_G^{k+1}v}}{\norm{M_G^{k}v}}
\times\frac{1+\zeta}{1-\zeta}\le (1-\lambda_2/2)
\times\frac{1+\zeta}{1-\zeta}
=\left(1+ \frac{2\zeta}{1-\zeta}\right)(1-\lambda_2/2)
\end{align*}
and thus
\[
2\left(1-\frac{C_2}{C_1}\right)\ge \lambda_2 -\frac{4\zeta}{1-\zeta}+\frac{2\zeta\lambda_2}{1-\zeta}
\ge
\left(1-\frac{4\zeta}{\lambda_2(1-\zeta)}\right)\lambda_2\ge
\left(1-\frac{5\zeta}{\lambda_2}\right)\lambda_2\ge\left(1-\delta\right)\lambda_2,
\]
where the third inequality was obtained from $\zeta\le 1/5$ (from the assumption $\zeta\le \delta\lambda_2/12$).

Let us now prove part (ii). For any unit vector $v$, H{\" o}lder's inequality gives  
\[
\norm{M_G^{k}v}^2\le \left(\sum_i v_i^2\lambda_i^{2k+2}\right)^{k/(k+1)}\left(\sum_i v_i^2\right)^{1/(k+1)}=
\norm{M_G^{k+1}v}^{2k/(k+1)}.
\]
If $v$ is good then
\[
\frac{\snorm{M_G^{k+1}v}}{\snorm{M_G^{k}v}}\ge
\frac{\snorm{M_G^{k+1}v}}{\snorm{M_G^{k+1}v}^{k/(k+1)}}=\snorm{M_G^{k+1}v}^{1/(k+1)}
\ge
\frac{(1-\lambda_2/2)}{(\sqrt{2}n\dmaxmin)^{1/(k+1)}}.
\]
Taking $k+1\ge \frac{3\log(\sqrt{2}n\dmaxmin)}{\delta\lambda_2}$ gives
$
\frac{1}{(\sqrt{2}n\dmaxmin)^{1/(k+1)}}\ge e^{-\delta\lambda_2/3}\ge 1-\delta\lambda_2/3.
$
Let $C_1$ and $C_2$ be as in the statement of the proposition. We have
\begin{align*}
\frac{C_2}{C_1}\ge \frac{\snorm{M_G^{k+1}v}}{\snorm{M_G^{k}v}}
\times\frac{1-\zeta}{1+\zeta}&\ge(1-\lambda_2/2)(1-\delta\lambda_2/3)
\times\left(1-2\zeta\right)\\
&\ge 1-2\zeta -\frac{\delta\lambda_2}{3}+\frac{2\zeta\delta\lambda_2}{3}-\frac{\lambda_2}{2}+\lambda_2\zeta+\frac{\delta\lambda_2^2}{3}\left(\frac{1}{2}-\zeta\right)\\
&\ge 
1-2\zeta -\frac{\delta\lambda_2}{3}-\frac{\lambda_2}{2},
\end{align*}
where the last inequality uses $\zeta\le 1/2$ (which is guaranteed from the assumption $\zeta\le \delta\lambda_2/12$). 
We thus obtain 
$
2\left(1-C_2/C_1\right)\le \left(1+2\delta/3+4\zeta/\lambda_2\right)\lambda_2\le (1+\delta)\lambda_2,
$
as claimed.
\end{proof}

Proposition \ref{proposition:power} requires good multiplicative approximations of $\snorm{M_G^k v}$ and $\snorm{M_G^{k+1} v}$ to approximate $\lambda_2$. Using random walks, we nevertheless will only be able to obtain additive approximations. To convert additive approximations into a good multiplicative approximations, we will need lower bounds on these two quantities. We will also need upper bounds in order to control the running time (and the space complexity) of our algorithm.
The following lemma shows the bounds we will use. 

\begin{lemma}\label{lemma:norm}
Let $\tau$ be any real number such that $\tau\in(0,1]$. For any unit vector $v\in \ImL$, $\snorm{M_G^k v}< \tau$ for all integers $k> \frac{2\log (1/\tau)}{\lambda_2}$. Additionally, if $v$ is good and $n\ge 4$ then $\snorm{M_G^{k+1} v}\ge 2\tau$ for all integers 
\begin{equation}\label{ineq2}
k\le
\frac{\log(1/\tau)-\log(2\sqrt{2}n\dmaxmin)}{2\lambda_2}-1.
\end{equation}
\end{lemma}
\begin{proof}
For any unit vector $v\in\ImL$ we have
\[
\norm{M_G^kv}\le (1-\lambda_2/2)^k\le e^{-k\lambda_2/2},
\]
which is upper bounded by $\tau$ for $k> \frac{2\log (1/\tau)}{\lambda_2}$.
If $v$ is good we further have
\[
\norm{M_G^{k+1}v}\ge \frac{(1-\lambda_2/2)^{k+1}}{\sqrt{2}n\dmaxmin}\ge
\frac{1}{\sqrt{2}n\dmaxmin}\left(e^{-1}-\frac{\lambda_2}{4} \right)^{(k+1)\lambda_2/2},
\]
where we used the formula 
$
e^{-1}-\frac{1}{2a}\le \left(1-1/a\right)^a
$
valid for any $a\ge 1$ (see, e.g., \cite{Mitrinovic70}) with $a=2/\lambda_2$.
Note that 
\[
e^{-1}-\lambda_2/4\ge e^{-1}-\frac{n}{4(n-1)}>\frac{3}{100}
\]
for $n\ge 4$.
We get $\norm{M_G^{k+1}v}\ge 2\tau$ whenever
\[k+1
\le
\frac{2}{\lambda_2}\times \frac{\log (2\sqrt{2}\tau n\dmaxmin)}{\log(3/100)}
= 
\frac{2}{\log(100/3)}\times\frac{\log(1/\tau)-\log(2\sqrt{2}n\dmaxmin)}{\lambda_2}.
\]
Finally, note that $2/\log(100/3)>1/2$.
\end{proof}

In order to use the the theory developed above,
we need to be able to estimate $\norm{M_G^kv}$ space-efficiently. This can be done using an approach based on quantum walks, similarly to what we did in Section \ref{sec:main}. The description the procedure based on this idea, and its analysis, are given in the appendix. This procedure, denoted $\proc{Estimate-norm}(G,k,v,\epsilon,\gamma)$, computes an $\epsilon$-additive approximation of $\norm{M_G^kv}$ with probability at least $1-\gamma$, for any $k\ge 0$, any $v\in\Sigma$ and any $\epsilon,\gamma\in(0,1]$.
We state the main result of the appendix as the following theorem.

\begin{theorem}\label{th:approximation}
Procedure $\proc{Estimate-norm}(G,k,v,\epsilon,\gamma)$ uses $O(\log (nk\dmaxmin/\epsilon)+\log\log(1/\gamma))$ space and outputs an $\epsilon$-additive approximation of $\norm{M_G^kv}$ with probability at least $1-\gamma$.
\end{theorem}

Our algorithm for estimating $\lambda_2$, denoted Algorithm $\mathcal{B}$, is given in Figure \ref{fig:algorithmB} and analyzed in Theorem \ref{th:l2} below. 

\begin{figure}[ht!]
\begin{center}
\fbox{
\begin{minipage}{15.8 cm} \vspace{-4mm}
\begin{codebox}
\zi
Input: an undirected, weighted and connected graph $G$ on $n$ vertices, where $n\ge 4$,\\

\hspace{11mm} a precision parameter $\delta\in(0,1]$, an error parameter $\gamma\in(0,1]$
 
\li
$\tau\leftarrow 
\frac{1}{2(\sqrt{2}n\dmaxmin)^{1+8/\delta}}$;
$\epsilon\leftarrow \frac{\delta\lambda\tau}{12}$; 
$\zeta\leftarrow \frac{4\gamma}{n(n-1)}\times\left(1+\frac{\log(1/\tau)}{\lambda_2}\right)^{-1}$;
\li
$R_{\max} \leftarrow 0$;
\li
\For all $v\in \Sigma$ \Do
\li
$k\leftarrow 1$; 
\li
$C_1\leftarrow \proc{Estimate-norm}(G,1,v,\epsilon,\zeta)$; 
\hspace{4mm}\# $C_1$ will store an approximation of $\norm{M_G^{k}v}$
\li
$C_2\leftarrow \proc{Estimate-norm}(G,2,v,\epsilon,\zeta)$; 
\hspace{4mm}\# $C_2$ will store an approximation of $\norm{M_G^{k+1}v}$
\li
\Repeat\hspace{3mm} until $C_2<3\tau/2$
\li
\If $C_2/C_1>R_{\max}$ 
\li
\Then $R_{\max}\leftarrow C_2/C_1$;
\End
\li
$k\leftarrow k +1$;
\li
$C_1\leftarrow C_2$;
\li
$C_2\leftarrow \proc{Estimate-norm}(G,k+1,v,\epsilon,\zeta)$;
\End
\End
\li
Output $2(1-R_{\max})$.
\end{codebox}
\end{minipage}
}
\end{center}\vspace{-4mm}
\caption{Algorithm $\mathcal{B}$ computing an $\delta$-multiplicative approximation of $\lambda_2$ with probability at least $1-\gamma$.}\label{fig:algorithmB}\vspace{-2mm}
\end{figure}

\begin{theorem}\label{th:l2}
Let $G$ be an undirected weighted connected graph and
$\lambda\in(0,2]$ be a lower bound on~$\lambda_2$.
Algorithm $\mathcal{B}$ outputs a $\delta$-multiplicative approximation of $\lambda_2$ with probability at least $1-\gamma$, and uses $O(\frac{1}{\delta}\log (n\dmaxmin/\lambda)+\log\log(1/\gamma))$ space.
\end{theorem}
\begin{proof}
Let us first analyze Algorithm $\mathcal{B}$ under the assumption that at Steps 5, 6 and 12, Procedure $\proc{Estimate-norm}(G,k,v,\epsilon,\zeta)$ always correctly outputs an $\epsilon$-additive approximation of $\snorm{M_G^kv}$. Then during the execution of the algorithm, $C_1$ and $C_2$ are $\epsilon$-additive approximations of $\snorm{M_G^kv}$ and $\snorm{M_G^{k+1}v}$, respectively. 

Lemma~\ref{lemma:norm} guarantees that for each $v\in\Sigma$, 
the inequality $\snorm{M_G^kv}< \tau$ holds for all integers $k> 2\log(1/\tau)/\lambda_2$, in which case we have $C_1<\tau+\epsilon<3\tau/2$ (the same inequality holds for $C_2$) since $\epsilon<\tau/2$. For each $v\in\Sigma$, the loop of Steps 7-12 is thus repeated at most $2\log(1/\tau)/\lambda_2$ times.

Whenever the ratio $C_2/C_1$ is computed at Step 8-9, we have $C_1>3\tau/2$ and $C_2>3\tau/2$. Since $\epsilon<\tau/2$, this means that $\snorm{M_G^kv}>\tau$ and $\snorm{M_G^{k+1}v}>\tau$. In this case the quantities $C_1$ and $C_2$ are thus also $\frac{\epsilon}{\tau}$-multiplicative approximations of $\snorm{M_G^{k}v}$ and $\snorm{M_G^{k+1}v}$, respectively. From part (i) of Proposition \ref{proposition:power} with $\zeta=\epsilon/\tau$, we conclude that 
$
2(1-R_{\max})\ge (1-\delta)\lambda_2.
$

Let $v$ be a good vector.
Observe that with the choice of $\tau$ made at Step 1 we have 
\[
\frac{\log(1/\tau)-\log(2\sqrt{2}n\dmaxmin)}{2\lambda_2}= 
\frac{4\log(\sqrt{2}n\dmaxmin)}{\delta\lambda_2}\ge 
\frac{3\log(\sqrt{2}n\dmaxmin)}{\delta\lambda_2}+1
\]
for $n\ge 4$.
There thus necessarily exists at least one integer $k$ satisfying both the condition of Part~(ii) of Proposition \ref{proposition:power} 
and Inequality~(\ref{ineq2}). 
From Lemma~\ref{lemma:norm}, for such a $k$ we have $\snorm{M_G^{k+1} v}\ge 2\tau$, which gives the lower bound $C_2\ge 2\tau-\epsilon\ge 3\tau/2$ (the same inequality holds for $C_1$), and implies that the ratio $C_2/C_1$ is computed at Steps 8-9. Part (ii) of Proposition \ref{proposition:power} and Lemma \ref{lemma:good} thus imply
$
2(1-R_{\max})\le (1+\delta)\lambda_2,
$
as claimed.

From Theorem \ref{th:approximation}, each call of Procedure $\proc{Estimate-norm}$ requires at most 
\[
O(\log (nk\dmaxmin/\epsilon)+\log\log(1/\zeta))=O(\frac{1}{\delta}\log (n\dmaxmin/\lambda)+\log\log(1/\gamma))
\]
space. This bound is also an upper bound on the space complexity of all the other computational steps of Algorithm $\mathcal{B}$.
Let us conclude by discussing the success probability of this algorithm. Each application of $\proc{Estimate-norm}$ errs with probability at most $\zeta$, from Theorem \ref{th:approximation}. There are at most $2+2\log(1/\tau)/\lambda_2$ calls to this procedure. The success probability is thus at least
$
1-|\Sigma|\left(2+2\log(1/\tau)/\lambda\right)\zeta\ge 1-\gamma.
$
\end{proof}
Theorem \ref{th:l2} implies Theorem \ref{th:main2} by observing again that $\dmaxmin/\lambda$ can be upper bounded by a polynomial in $n$ when the weights are polynomially bounded (if $G$ is not connected then we apply Theorem \ref{th:l2} on each connected component and taking the minimum of the estimations obtained).

\section*{Acknowkedgments}
The author is grateful to Richard Cleve, Hirotada Kobayashi, Harumichi Nishimura, Suguru Tamaki and Ryan Williams for helpful comments. This work is supported by the Grant-in-Aid for Young Scientists~(A) No.~16H05853, the Grant-in-Aid for Scientific Research~(A) No.~16H01705, and the Grant-in-Aid for Scientific Research on Innovative Areas~No.~24106009 of the Japan Society for the Promotion of Science and the Ministry of Education, Culture, Sports, Science and Technology in Japan. 


\begin{thebibliography}{10}

\bibitem{Berkovitz84}
Stuart~J. Berkovitz.
\newblock On computing the determinant in small parallel time using a small
  number of processors.
\newblock {\em Information Processing Letters}, pages 147--150, 1984.

\bibitem{Borodin+82}
Allan Borodin, Joachim von~zur Gathen, and John~E. Hopcroft.
\newblock Fast parallel matrix and {GCD} computations.
\newblock {\em Information and Control}, 52(3):241--256, 1982.

\bibitem{Chung97}
Fan R.~K. Chung.
\newblock {\em Spectral Graph Theory}.
\newblock American Mathematical Society, 1997.

\bibitem{Chung+15}
Fan R.~K. Chung and Olivia Simpson.
\newblock Solving local linear systems with boundary conditions using heat
  kernel {Pagerank}.
\newblock {\em Internet Mathematics}, 11:4--5, 2015.

\bibitem{Csanky+FOCS75}
Laszlo Csanky.
\newblock Fast parallel matrix inversion algorithms.
\newblock In {\em Proceedings of the 16th Annual Symposium on Foundations of
  Computer Science}, pages 11--12, 1975.

\bibitem{Doron+ICALP15}
Dean Doron and Amnon Ta{-}Shma.
\newblock On the problem of approximating the eigenvalues of undirected graphs
  in probabilistic logspace.
\newblock In {\em Proceedings of the 42nd International Colloquium}, pages
  419--431, 2015.

\bibitem{Glynn87}
Peter~W. Glynn.
\newblock Upper bounds on {Poisson} tail probabilities.
\newblock {\em Operation Research Letters}, 6(1):9--14, 1987.

\bibitem{Kelner+STOC13}
Jonathan~A. Kelner, Lorenzo Orecchia, Aaron Sidford, and Zeyuan~Allen Zhu.
\newblock A simple, combinatorial algorithm for solving {SDD} systems in
  nearly-linear time.
\newblock In {\em Proceedings of the 45th Symposium on Theory of Computing},
  pages 911--920, 2013.

\bibitem{Mitrinovic70}
Dragoslav Mitrinovi{\' c}.
\newblock {\em Analytic Inequalities}.
\newblock Springer, 1970.

\bibitem{Mitzenmacher+05}
Michael Mitzenmacher and Eli Upfal.
\newblock {\em Probability and Computing}.
\newblock Cambridge University Press, 2005.

\bibitem{Peng+STOC14}
Richard Peng and Daniel~A. Spielman.
\newblock An efficient parallel solver for {SDD} linear systems.
\newblock In {\em Proceedings of the 46th Symposium on Theory of Computing},
  pages 333--342, 2014.

\bibitem{Spielman+STOC04}
Daniel~A. Spielman and Shang{-}Hua Teng.
\newblock Nearly-linear time algorithms for graph partitioning, graph
  sparsification, and solving linear systems.
\newblock In {\em Proceedings of the 36th Annual Symposium on Theory of
  Computing}, pages 81--90, 2004.

\bibitem{Spielman+11}
Daniel~A. Spielman and Shang{-}Hua Teng.
\newblock Spectral sparsification of graphs.
\newblock {\em {SIAM} Journal on Computing}, 40(4):981--1025, 2011.

\bibitem{Spielman+13}
Daniel~A. Spielman and Shang{-}Hua Teng.
\newblock A local clustering algorithm for massive graphs and its application
  to nearly linear time graph partitioning.
\newblock {\em {SIAM} Journal on Computing}, 42(1):1--26, 2013.

\bibitem{Spielman+14}
Daniel~A. Spielman and Shang{-}Hua Teng.
\newblock Nearly linear time algorithms for preconditioning and solving
  symmetric, diagonally dominant linear systems.
\newblock {\em {SIAM} Journal on Matrix Analysis and Applications},
  35(3):835--885, 2014.

\bibitem{TaShmaSTOC13}
Amnon Ta{-}Shma.
\newblock Inverting well conditioned matrices in quantum logspace.
\newblock In {\em Proceedings of the 45th Symposium on Theory of Computing},
  pages 881--890, 2013.

\bibitem{Vishnoi13}
Nisheeth~K. Vishnoi.
\newblock {\em $Lx=b$ --- Laplacian Solvers and their Algorithmic
  Applications}.
\newblock Now publishers, 2013.

\end{thebibliography}

\section*{Appendix: Estimating $\norm{M_G^k v}$}
In this appendix we explain how to space-efficiently estimate $\norm{M_G^k v}$ for any vector $v\in\Sigma$, and prove Theorem \ref{th:approximation}.

The following lemma first shows how to space-efficiently estimate the quantity $2^{-k}{k \choose s}$. 
\begin{lemma}\label{lemma:binomial}
Let $\delta$ and $\zeta$ be any real numbers such that $\delta\in(0,1)$ and $\zeta>0$. There exists a $O(\log{(ks/\delta)+\log\log(1/\zeta)})$-space algorithm that, when given as input two integers $k\ge 1$ and $s\in\{0,\ldots k\}$, outputs an $\delta$-additive approximation of $2^{-k}{k \choose s}$ with probability $1-\zeta$.
\end{lemma}
\begin{proof}
Our algorithm is as follows. The algorithm creates a counter $C$ initialized to zero, and repeat $m=\ceil{\frac{\log(2/\zeta)}{2\delta^2}}$ times the following: take $s$ bits uniformly at random and increment $C$ by one if exactly $k$ among these $s$ bits are one. The algorithm finally outputs $C/m$. Observe that this algorithm can be implemented in space linear in $\log(1/\delta)$, $\log k$, $\log s$ and $\log\log (1/\zeta)$. 

The expected value of the output of this algorithm is precisely $2^{-k}{k \choose s}$. From Chernoff bound, the output of the algorithm is thus a $\delta$-additive approximation of this quantity with probability at least $1-2e^{-2m\delta^2}\ge 1-\zeta$. 
\end{proof}

The procedure $\proc{Estimate-norm}(G,k,v,\epsilon,\gamma)$ estimating $\norm{M_G^k v}$ is described in Figure \ref{fig:algorithmC}.

\begin{figure}[ht!]
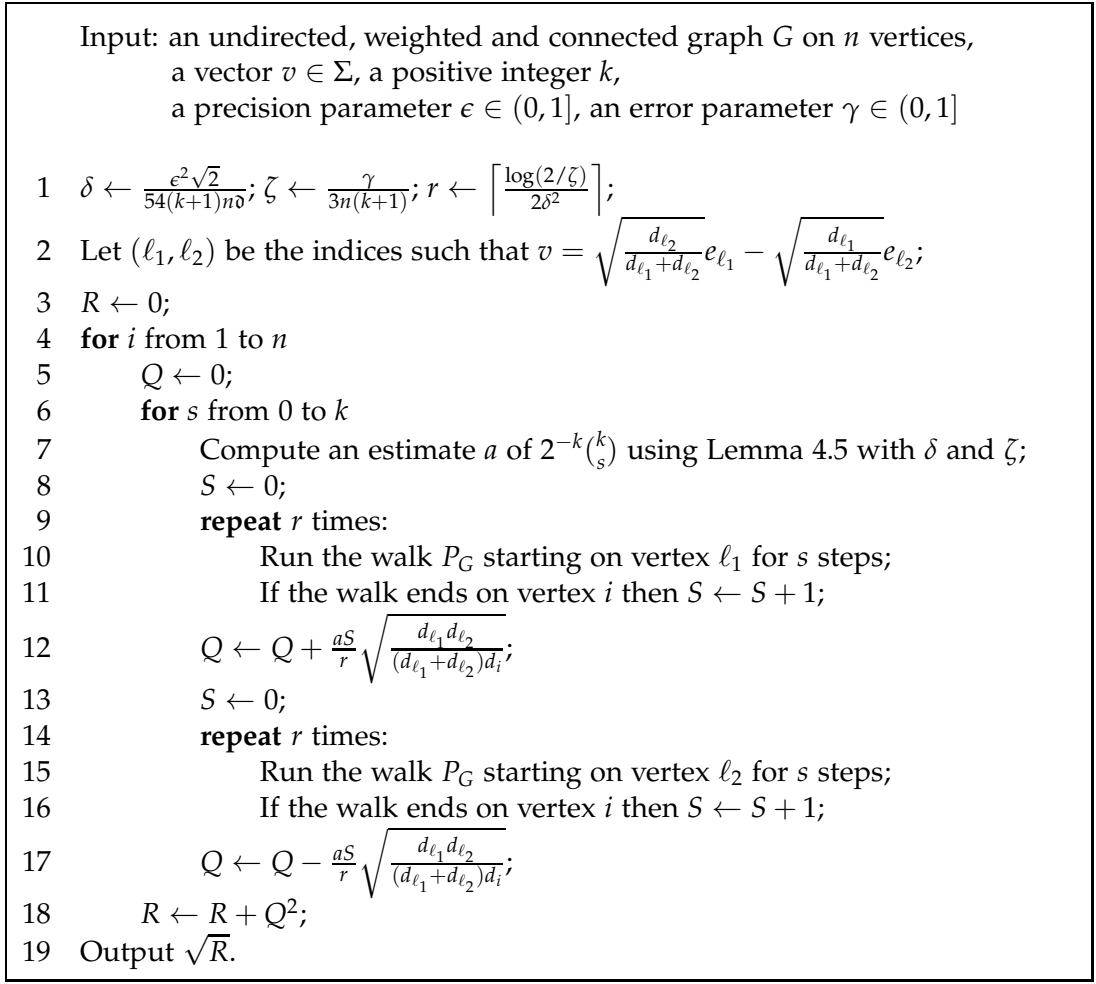

\begin{center}
\fbox{
\begin{minipage}{14 cm} \vspace{-4mm}
\begin{codebox}
\zi
Input: an undirected, weighted and connected graph $G$ on $n$ vertices, 
\zi
\hspace{11mm} a vector $v\in\Sigma$, a positive integer $k$,
\zi
\hspace{11mm}  a precision parameter $\epsilon\in(0,1]$, an error parameter $\gamma\in(0,1]$
\zi
\li
$\delta\leftarrow \frac{\epsilon^2\sqrt{2}}{54(k+1)n\dmaxmin}$; 
$\zeta\leftarrow \frac{\gamma}{3n(k+1)}$; $r\leftarrow \ceil{\frac{\log(2/\zeta)}{2\delta^2}}$;
\li
Let $(\ell_1,\ell_2)$ be the indices such that 
$v=\sqrt{\frac{d_{\ell_2}}{d_{\ell_1}+d_{\ell_2}}}e_{\ell_1}-\sqrt{\frac{d_{\ell_1}}{d_{\ell_1}+d_{\ell_2}}}e_{\ell_2}$;
\li
$R\leftarrow 0$; 
\li
\For $i$ from 1 to $n$ \Do
\li
$Q\leftarrow 0$;
\li
\For $s$ from 0 to $k$ \Do
\li
Compute an estimate $a$ of $2^{-k}{k\choose s}$ using Lemma \ref{lemma:binomial} with $\delta$ and $\zeta$;
\li
$S\leftarrow 0$;
\li
\Repeat\hspace{3mm} $r$ times:
\li
Run the walk $P_G$ starting on vertex $\ell_1$ for $s$ steps;
\li
If the walk ends on vertex $i$ then
$S\leftarrow S+1$;
\End
\li
$Q\leftarrow Q+\frac{aS}{r}\sqrt{\frac{d_{\ell_1}d_{\ell_2}}{(d_{\ell_1}+d_{\ell_2})d_i}}$;
\li
$S\leftarrow 0$;
\li
\Repeat\hspace{3mm} $r$ times:
\li
Run the walk $P_G$ starting on vertex $\ell_2$ for $s$ steps;
\li
If the walk ends on vertex $i$ then
$S\leftarrow S+1$;
\End
\li
$Q\leftarrow Q-\frac{aS}{r}\sqrt{\frac{d_{\ell_1}d_{\ell_2}}{(d_{\ell_1}+d_{\ell_2})d_i}}$;
\End
\li
$R\leftarrow R+Q^2$;
\End
\li 
Output $\sqrt{R}$.
\end{codebox}
\end{minipage}
}
\end{center}\vspace{-4mm}
\caption{Procedure $\proc{Estimate-norm}(G,k,v,\epsilon,\gamma)$ computing an $\epsilon$-additive approximation of $\norm{M_G^kv}$ with probability at least $1-\gamma$.}\label{fig:algorithmC}
\end{figure}

We now prove Theorem \ref{th:approximation} stated in Section \ref{sec:l2}.
\begin{proof}[Proof of Theorem \ref{th:approximation}]
Observe that
\[
M_G^k=\sum_{i=1}^n(1-\lambda_i/2)^k u_i^tu_i=D_G^{1/2}\left(
\sum_{s=0}^k \frac{{k \choose s}}{2^k}P_G^s\right)D_G^{-1/2}.
\]
Moreover $M^k_G$ is symmetric. Let us write 
$
v=\sqrt{\frac{d_{\ell_2}}{d_{\ell_1}+d_{\ell_2}}}e_{\ell_1}-\sqrt{\frac{d_{\ell_1}}{d_{\ell_1}+d_{\ell_2}}}e_{\ell_2}
$
as in Step 2 of the procedure.
We thus have
\begin{align*}
\norm{M_G^k v}&=\norm{ v^t D_G^{1/2}\left(
\sum_{s=0}^k \frac{{k \choose s}}{2^k}P_G^s\right)D_G^{-1/2}}
=
\sqrt{\sum_{i=1}^n
\Gamma_i^2}
\end{align*}
where
\[
\Gamma_i=\sum_{s=0}^k\frac{{k \choose s}}{2^{k}} \left( 
\sqrt{\frac{d_{\ell_1}d_{\ell_2}}{(d_{\ell_1}+d_{\ell_2})d_i}}\braket{e_i,e_{\ell_1} P_G^s}
-
\sqrt{\frac{d_{\ell_1}d_{\ell_2}}{(d_{\ell_1}+d_{\ell_2})d_i}}\braket{e_i,e_{\ell_2} P_G^s}\right).
\]
Note that $|\Gamma_i|\le 1$ for any $i\in\{1,\ldots,n\}$.

At the end of Steps 9-11 we have
\begin{equation}\label{eq:Chernoff2}
\Pr\left[|S/r-\langle e_i,e_{\ell_1} P_G^s\rangle|\le \delta\right]\ge 1-2e^{-2r\delta^2}\ge 1-\zeta,
\end{equation}
using the same argument as in the analysis in Theorem \ref{th:algA}. The same bound, with $\ell_1$ replaced by $\ell_2$, holds at the end of Steps 14-16 as well.
Lemma \ref{lemma:binomial} also shows that $a$ is a $\delta$-additive approximation of $2^{-k}{k\choose s}$ with probability at least $1-\zeta$ at Step 7. Let us continue our analysis under the assumption that all these approximations are correct (we discuss the overall success probability at the end of the proof). At Step 12 we thus have
\[
\left|
\frac{a S}{r} -
\frac{{k \choose s}}{2^{k}}\braket{e_i,e_{\ell_1} P_G^s}
\right|\le 
\left(\delta+\frac{{k \choose s}}{2^{k}}+\braket{e_i,e_{\ell_1}P_G^s}\right)\delta\le 3\delta,
\]
and the same bound, with $\ell_1$ replaced by $\ell_2$, holds at Step 17 as well.
For any $i\in\{1,\ldots,n\}$ we thus have
\[
\left|
Q-
\Gamma_i
\right|\le 6(k+1)\delta \max_{\ell_1,\ell_2\in\{1,\ldots,n\}} \sqrt{\frac{d_{\ell_1}d_{\ell_2}}{(d_{\ell_1}+d_{\ell_2})d_i}}\le \frac{6}{\sqrt{2}}(k+1)\delta \dmaxmin.
\]
at the end of the loop of Steps 5-17, which implies
\[
\left|
Q^2-\Gamma_i^2
\right| = 
\left|
Q-\Gamma_i
\right|\times \left|
Q+\Gamma_i
\right|\le
\frac{18}{\sqrt{2}}(k+1)\delta \dmaxmin.
\]
since $|\Gamma_i|\le 1$ and $|Q|\le |Q-\Gamma_i|+|\Gamma_i|\le 2$. We thus have
\[
\left|
R -
\norm{M_G^k v}^2
\right|
\le \frac{18}{\sqrt{2}}(k+1)n\delta \dmaxmin\le \epsilon^2/3.
\]
at the end of of the algorithm. Let us show that $\sqrt{R}$ is an $\epsilon$-additive estimation of $\norm{M_G^kv}$ by considering two cases.
In the case $\sqrt{|R|}\le \epsilon/3$ we get
\[
\left|
\sqrt{R} -
\norm{M_G^k v}
\right|\le 
\sqrt{R} + \norm{M_G^kv}
\le
\sqrt{R} + \sqrt{R+\left|
R -
\norm{M_G^k v}^2
\right|}
\le 
\epsilon/3 +\sqrt{\epsilon^2/9+ \epsilon^2/3}=\epsilon.
\]
Now in the case $\sqrt{|R|}\ge \epsilon/3$ we get
\[
\left|
\sqrt{R} -
\norm{M_G^k v}
\right|=
\frac{\left|
R -
\norm{M_G^k v}^2
\right| }{\left|
\sqrt{R} +
\norm{M_G^k v}
\right|}
\le \frac{\epsilon^2/3}{\epsilon/3}=\epsilon.
\]

The space complexity of this algorithm is $O(\log(nk/\epsilon)+\log\log \gamma)$, from Lemma \ref{lemma:binomial} and the observation that all other computational steps can be implemented with registers of this size.
 Let us conclude by discussing the success probability of this algorithm. Errors can only occur at Steps 6 or 10. From Inequality (\ref{eq:Chernoff2}) and Lemma \ref{lemma:binomial}, and using the union bound, we know that the overall success probability is at least $1-3\zeta n(k+1)= 1-\gamma$.
\end{proof}

\end{document}